\documentclass[11pt]{article}

\usepackage{graphicx,color,epsfig}
\usepackage{amssymb, amsmath, amsthm, amsfonts, amscd}
\usepackage{color}

\theoremstyle{plain}% default
\newtheorem{thm}{Theorem}[section]
\newtheorem{lem}[thm]{Lemma}
\newtheorem{prop}[thm]{Proposition}

\theoremstyle{definition}

\newtheorem{exmp}{Example}[section]
\theoremstyle{remark}
\newtheorem{rem}{Remark}[section]

\numberwithin{equation}{section}
\numberwithin{figure}{section}

\def\Real{{\mathbb{R}}}
\def\Complex{{\mathbb{C}}}

\def\Re{\mathop{\rm Re}\nolimits}
\def\Im{\mathop{\rm Im}\nolimits}

\begin{document}
\title{The Laplace Equation in the Exterior of the Hankel Contour and Novel Identities for Hypergeometric Functions}
\author{A.S. Fokas$^\dagger$ and M.L. Glasser$^\ddagger$ \\
{\small $^\dagger$Department of Applied Mathematics and Theoretical Physics}\\
{\small University of Cambridge, CB3 0WA, UK and} \\
{\small Research Center of Mathematics, Academy of Athens, Greece}\\
{\small $^\ddagger$Department of Physics, Clarkson University, Potsdam, NY 13699-5820, USA}
}

\maketitle

\begin{abstract}
By employing conformal mappings, it is possible to express the solution of certain boundary value problems for the Laplace equation in terms of a single integral involving the given boundary data. We show that such explicit formulae can be used to obtain novel identities for special functions. A convenient tool for deriving this type of identities is the so-called \emph{global relation}, which has appeared recently in a wide range of boundary value problems. As a concrete application, we analyze the Neumann boundary value problem for the Laplace equation in the exterior of the so-called Hankel contour, which is the contour that appears in the definition of both the gamma and the Riemann zeta functions. By utilizing the explicit solution of this problem, we derive a plethora of novel identities involving the hypergeometric function.
\end{abstract}

\section{Introduction}

Simple boundary value problems for the Laplace equation in two dimensions can be solved by the powerful technique of conformal mappings. However, the usual implementation of the conformal mappings presented in most text books is limited, because it applies only to the case that the given boundary data is piecewise constant. In section \ref{s:2} we present a technique which can be applied to the general case of arbitrary boundary data (see also \cite{Henricibook}). By employing this technique, we construct in section \ref{s:2} a solution of the Neumann boundary value of the Laplace equation in the domain $D$, which is the exterior of the so-called Hankel contour $H$ (see figure \ref{Hankel}) defined by
\begin{equation}
\label{H}
H=\{ z=r e^{-i\alpha}, ~a<r<\infty\}\cup \{ z=a e^{-i\theta}, ~-\alpha<\theta<\alpha\}\cup
\{ z=r e^{i\alpha}, ~a<r<\infty\},
\end{equation}
where
\begin{equation}
 0<a<2\pi, \qquad |\alpha|<\pi.
\nonumber
\end{equation}
A solution of the above Neumann boundary value problem is expressed in Theorem \ref{thm2.1} in terms of a single integral involving the three functions $\{g_+(r), g(\varphi), g_-(r)\}$ defining the Neumann data. Using this explicit solution, we derive in section \ref{s:3} four integral identities. A particular case of these four identities can be derived by employing the explicit solutions
\begin{equation}
 q(r,\theta)=\Re (z^k), \quad q(r,\theta)=\Im (z^k), \qquad z=re^{i\theta},
\quad r>0, \quad \theta\in\Real, \quad \Re k>0.
\end{equation}
In the general case, these identities are derived using the so-called \emph{global relation}. We recall that this relation plays a crucial role for the contruction of the solution of a wide class of initial-boundary value problems; for evolution and elliptic PDEs respectively, see for example \cite{nlty21pT195,prsla465p3341,prs464p1823,ima67p1,ima70p1,prsa461p2965,mpcps136p361} and \cite{mpcps138p339,prsa2012,
pre64p016114,prsla460p1285,prsla461p2721,prsla457p371,qjmam55p457,%
prsa466p2259,prsa466p2283}. Here, instead of employing the global relation to solve the given boundary value problem, having already constructed the solution via conformal mappings, we use the global relation to obtain the integral identities mentioned above.

Regarding the global relation, we note that in addition to its basic role for the analytical solution of a large class of PDEs \cite{prssa453p1411,jmp41p4188,ASFbook}, it has also been utilized in the following contexts: $(a)$ It yields a novel non-local formulation of the classical problem of water waves \cite{jfm562p313,jfm2011,jfm675p141,jfm2012ASF,jfm631p375}; $(b)$ it provides a useful approach to Hele-Shaw type problems \cite{nlty21pT205}; $(c)$ it gives rise to novel numerical techniques for elliptic PDEs in the interior of a convex polygon \cite{ASFEAS2009,prsa467p2983,jcam236p2515,jcam219p9,ima30p1184}.

%\begin{figure}[h]
%\centerline{\scalebox{0.6}{\input{contourH.pstex_t}}}
%\caption{The Hankel contour.}
%\label{Hankel}
%\end{figure}
%
%%%%%%%%%%%%%%%%%%%%%%%%%%%%%%%%%%%%%%%%%%%%%%%%%%%%%%%%%%%%%%%%%%%%%%%%%%%%%%%%%%%%%
\begin{figure}[t!]
\kern\medskipamount
\centerline{
\includegraphics[width=0.335\textwidth]{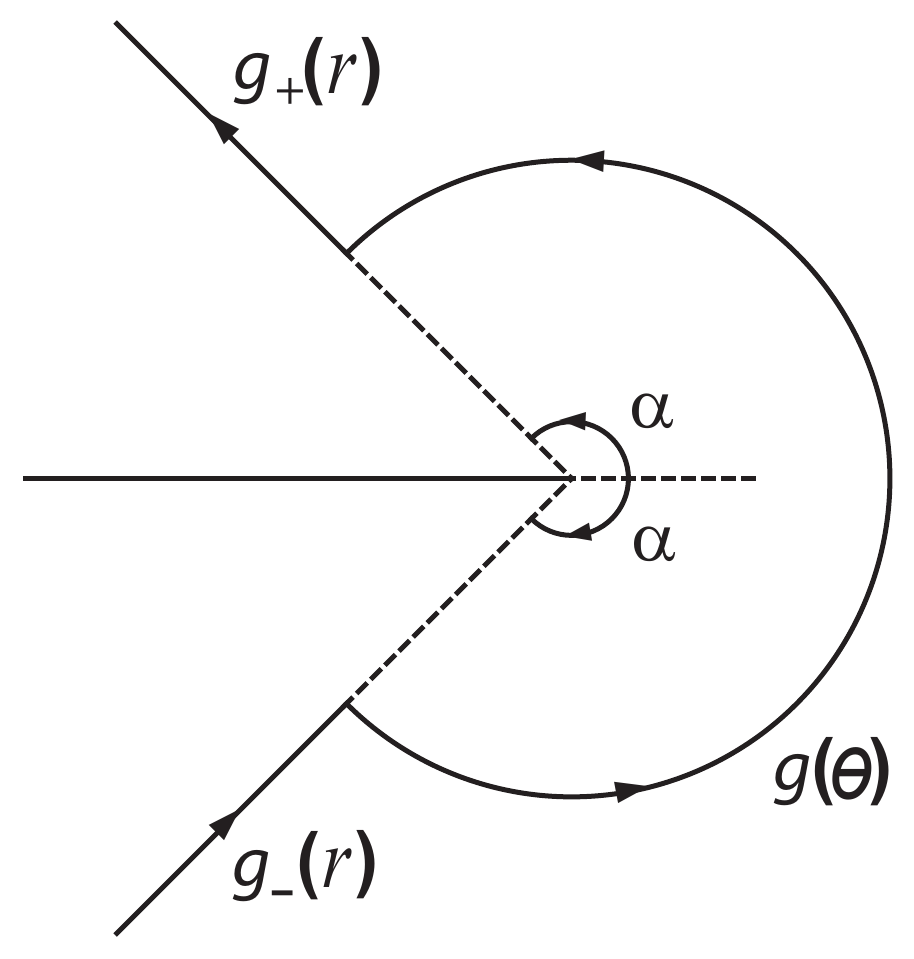}}
\caption{ The Hankel contour.
}
\label{Hankel}
\vskip3\medskipamount
\end{figure}
%%%%%%%%%%%%%%%%%%%%%%%%%%%%%%%%%%%%%%%%%%%%%%%%%%%%%%%%%%%%%%%%%%%%%%%%%%%%%%%%%%%%%
%

\section{A Neumann Boundary Value Problem for the Laplace Equation}
\label{s:2}
The Dirichlet and Neumann boundary value problems for the Laplace equation in the exterior of the Hankel contour can be solved via conformal mappings. In this respect we note that the usual implementation of the conformal mappings presented in most text books fails here, because it only applies to the case of piecewise constant Dirichlet or Neumann data. For the case of arbitrary Neumann data the following result is useful.

\begin{prop}
 Let $q(r,\theta)$ satisfy a Neumann boundary value problem for the Laplace equation. Let $\omega(z)$ be
the conformal mapping of the relevant domain to the upper half of the complex $\omega$-plane, where
\begin{equation}
\omega(z)=x+iy,\quad  z=re^{i\theta}, \quad  x,y,\theta\in \Real, \quad r>0.
\end{equation}
Then,
\begin{equation}
 q(r,\theta)=\frac1{2\pi}\int_{-\infty}^{\infty} \ln\big[(\xi-x)^2+y^2\big]q_y(\xi,0)d\xi,
\label{qrtheta}
\end{equation}
where
\begin{equation}
 x=\Re\omega(r,\theta), \qquad y=\Im\omega(r,\theta)
\end{equation}
and $q_y(\xi,0)$ can be computed in terms of the given Neumann data using the identity
\begin{equation}
 q_y(x,0)dx=\left[\frac1{r}q_\theta dr-rq_rd\theta\right]_{y=0}.
\label{qydx}
\end{equation}

\end{prop}

\begin{proof}
 Observe that
\begin{equation}
 q_y dx\big|_{y=0}=i\big[q_\omega d\omega-q_{\bar\omega}d\bar\omega\big]_{y=0}.
\label{qydxa}
\end{equation}
Indeed, the RHS of this equation equals
\begin{equation}
 \frac{i}2\big[(q_x-iq_y\big)dx-(q_x+iq_y)dx\big]_{y=0},
\nonumber
\end{equation}
which equals the LHS of \eqref{qydxa}.

The conformal mapping implies
\begin{equation}
 q_\omega d\omega=q_z dz.
\nonumber
\end{equation}
Furthermore, using $z=re^{i\theta}$, we find
\begin{align}
 q_z dz &= \left[ \frac{e^{-i\theta}}2\left(q_r-\frac{i}{r}q_\theta\right)\right]
\left[e^{i\theta}(dr+ird\theta)\right]
\nonumber\\
&=\frac{1}{2}\big[q_rdr+q_\theta d\theta\big]+\frac{i}{2}\left[ rq_rd\theta-\frac1{r}q_\theta dr\right].
\nonumber
\end{align}
Hence,
\begin{equation}
 q_\omega d\omega-q_{\bar \omega}d\bar\omega=i\left[rq_r d\theta-\frac1{r}q_\theta dr\right]
\nonumber
\end{equation}
and equation \eqref{qydxa} becomes \eqref{qydx}.

Equation \eqref{qrtheta} is the well known Poisson formula for the Neumann problem.
\qedhere

\end{proof}

\begin{thm}\label{thm2.1}
Let the domain $D$ { be} defined by
\begin{equation}
\label{D}
D=\{ a<r<\infty, \ -\alpha<\theta<\alpha, \}, \quad 0<a<2\pi, \ \frac{\pi}{2} <\alpha \leq  \pi.
\end{equation}
Let $q(r,\theta)$ solve the following Neumann boundary value problem for the Laplace equation in the domain $D$:
\begin{subequations}
\label{BCs}
\begin{equation}
q_\theta(r,-\alpha)=g_{-}(r), \quad a<r<\infty,
 \end{equation}
\begin{equation}
 q_r(a,\theta)=g(\theta), \quad -\alpha<\theta<\alpha,
\end{equation}
\begin{equation}
 q_\theta(r, \alpha)=g_{+}(r), \quad a<r<\infty,
\end{equation}
\end{subequations}
where the functions $g_{\pm}(r)$ and $g(\theta)$ have appropriate smoothness and decay.
A solution of this boundary value problem is given by
\begin{subequations}
\begin{equation}
\label{sol_Lap}
\begin{split}
q(r,\theta) &= -\frac{1}{2\pi}\int_a^{\infty}\frac{g_+(\rho)}{\rho}\ln\left\{\frac{1}{4}\left[\left(\frac{r}{a}\right)^{\frac{\pi}{2\alpha}}-\left(\frac{r}{a}\right)^{-\frac{\pi}{2\alpha}}\right]^2 \cos^2\left(\frac{\pi\theta}{2\alpha}\right)+\left[\frac{R(\rho)}{2}-\frac{R(r)}{2}\sin\left(\frac{\pi \theta}{2\alpha}\right)\right]^2\right\}d\rho \\
&+\frac{1}{2\pi}\int_a^{\infty}\frac{g_-(\rho)}{\rho}\ln\left\{\frac{1}{4}\left[\left(\frac{r}{a}\right)^{\frac{\pi}{2\alpha}}-\left(\frac{r}{a}\right)^{-\frac{\pi}{2\alpha}}\right]^2 \cos^2\left(\frac{\pi\theta}{2\alpha}\right)+\left[\frac{R(\rho)}{2}+\frac{R(r)}{2}\sin\left(\frac{\pi \theta}{2\alpha}\right)\right]^2\right\}d\rho \\
&+\frac{a}{2\pi}\int_{-\alpha}^{\alpha}g(\varphi)\ln\left\{\frac{1}{4}\left[\left(\frac{r}{a}\right)^{\frac{\pi}{2\alpha}}-\left(\frac{r}{a}\right)^{-\frac{\pi}{2\alpha}}\right]^2 \cos^2\left(\frac{\pi\theta}{2\alpha}\right)+\left[\sin\left(\frac{\pi\varphi}{2\alpha}\right)-\frac{R(r)}{2}\sin\left(\frac{\pi \theta}{2\alpha}\right)\right]^2\right\}d\varphi,
\end{split}
\end{equation}
where $R(\rho)$ is defined by
\begin{equation}
 R(\rho)=\left(\frac{\rho}{a}\right)^{\frac{\pi}{2\alpha}}
+\left(\frac{\rho}{a}\right)^{-\frac{\pi}{2\alpha}}
\end{equation}
\end{subequations}
and principal value integrals are assumed if needed.
\end{thm}
\begin{proof}
It is straightforward to verify that the function $\omega(z)$ defined by
\begin{equation}
 \omega(z)=\frac{i}2\left[ \left(\frac{z}{a}\right)^{\frac{\pi}{2\alpha}}
-\left(\frac{z}{a}\right)^{\frac{\pi}{-2\alpha}}\right],
\end{equation}
maps the domain $D$ defined in \eqref{D} to the upper half complex $\omega$-plane. The points
$\{ ae^{i\alpha}, ae^{-i\alpha}, a\}$ are mapped to the points $\{ -1,1,0\}$ respectively. Also
\begin{align}
 x &= -\frac12 \sin\left(\frac{\pi\theta}{2\alpha}\right)\left[ \left(\frac{r}{a}\right)^{\frac{\pi}{2\alpha}}
+\left(\frac{r}{a}\right)^{-\frac{\pi}{2\alpha}}\right]
\nonumber\\
y &=\frac12 \cos\left(\frac{\pi\theta}{2\alpha}\right)\left[ \left(\frac{r}{a}\right)^{\frac{\pi}{2\alpha}}
-\left(\frac{r}{a}\right)^{-\frac{\pi}{2\alpha}}\right].
\label{xny}
\end{align}
The condition $y=0$ implies either $\theta=\pm\alpha$ or $r=a$. Hence,
\begin{equation}
 q_y(\xi,0)d\xi=\left[\begin{aligned} &-aq_r(a,\varphi)d\varphi\\
                 &\frac1\rho q_\theta(\rho,\pm\alpha)d\rho
                \end{aligned}\right.  .
\nonumber
\end{equation}
Thus, equation \eqref{qrtheta} becomes
\begin{subequations}
\begin{align}
 q(r,\theta)=& \frac1{2\pi}\int_{\infty}^a \ln\tilde F(\rho,\alpha) q_\theta(\rho,\alpha)\frac{d\rho}{\rho}
+\frac1{2\pi}\int_{a}^{\infty} \ln\tilde F(\rho,-\alpha) q_\theta(\rho,-\alpha)\frac{d\rho}{\rho}
\nonumber\\
&+\frac1{2\pi}\int_{-\alpha}^{\alpha} \ln\tilde F(a,\varphi)\big[-a q_r(a,\varphi)\big]d\varphi,
\label{qrthetaint}
\end{align}
where
\begin{equation}
\tilde F(\rho,\varphi)=\big[\xi(\rho,\varphi)-x(r,\theta)\big]^2+y^2(r,\theta)
\end{equation}
\end{subequations}
and for convenience of notation we have suppressed the $(r,\theta)$ dependence of $\tilde F$.

Using
\begin{align}
 \xi(\rho,\alpha)&=-\frac12\sin\left(\frac{\pi}2\right)
\left[\left(\frac{\rho}{\alpha}\right)^{\frac{\pi}{2\alpha}}
+\left(\frac{\rho}{\alpha}\right)^{-\frac{\pi}{2\alpha}}\right]=-\frac12 R(\rho),
\nonumber\\
\xi(\rho,-\alpha)&=-\frac12\sin\left(-\frac{\pi}2\right)
\left[\left(\frac{\rho}{\alpha}\right)^{\frac{\pi}{2\alpha}}
+\left(\frac{\rho}{\alpha}\right)^{-\frac{\pi}{2\alpha}}\right]=\frac12 R(\rho),
\nonumber\\
{\xi(a,\varphi)}&=-\sin\left(\frac{\pi\varphi}{2\alpha}\right),
\nonumber
\end{align}
as well as the boundary conditions \eqref{BCs}, equation \eqref{qrthetaint} becomes the following equation:
\begin{align}
 q(r,\theta)&= -\frac{1}{2\pi}\int_a^{\infty} \frac{g_+(\rho)}{\rho}\ln\left[ \frac14 R^2(r)
-\cos^2\left(\frac{\pi\theta}{2\alpha}\right)+\frac14R^2(\rho)
-\frac12\sin\left(\frac{\pi\theta}{2\alpha}\right)R(\rho)R(r)\right]d\rho
\nonumber\\
&+\frac{1}{2\pi}\int_a^{\infty} \frac{g_-(\rho)}{\rho}\ln\left[ \frac14 R^2(r)
-\cos^2\left(\frac{\pi\theta}{2\alpha}\right)+\frac14R^2(\rho)
+\frac12\sin\left(\frac{\pi\theta}{2\alpha}\right)R(\rho)R(r)\right]d\rho
\nonumber\\
&+\frac{a}{2\pi}\int_{-\alpha}^{\alpha} g(\rho)\ln\left[ \frac14 R^2(r)
-\cos^2\left(\frac{\pi\theta}{2\alpha}\right)+\sin^2\left(\frac{\pi\varphi}{2\alpha}\right)
-\sin\left(\frac{\pi\varphi}{2\alpha}\right)\sin\left(\frac{\pi\theta}{2\alpha}\right)R(r)
\right]d\varphi.
\label{qrthetainta}
\end{align}

By employing the identity
\begin{equation}
 \left(\frac{R(r)}{2}\right)^2-1=\frac14\left[\left(\frac{r}{a}\right)^{\frac{\pi}{2\alpha}}
-\left(\frac{r}{a}\right)^{-\frac{\pi}{2\alpha}}\right]^2,
\nonumber
\end{equation}
equation \eqref{qrthetainta} becomes equation \eqref{sol_Lap}.
\qedhere

\end{proof}

\begin{rem}
(Asymptotics) The large $r$ asymptotics of the solution $q(r,\theta)$ defined by \eqref{sol_Lap} is given by
\begin{equation}
 q(r,\theta)=\frac1{2\pi}\left[ \ln\left(\frac14\right)
+\frac{\pi}{\alpha}\ln\left(\frac{r}{a}\right)\right] S
+\frac1{\pi}\sin\left(\frac{\pi\theta}{2\alpha}\right)\left(\frac{r}{a}\right)^{-\frac{\pi}{2\alpha}}\tilde S
+{O\left[\left(\frac{r}{a}\right)^{-\frac{2\pi}{\alpha}}\right]}, \quad r\to\infty,
\label{qasympt}
\end{equation}
where
\begin{equation}
 S=\int_a^\infty \big[g_-(\rho)-g_+(\rho)\big]\frac{d\rho}{\rho}+a\int_{-\alpha}^\alpha g(\varphi)d\varphi
\label{defS}
\end{equation}
and
\begin{equation}
 \tilde S=\int_a^\infty \big[g_+(\rho)+g_-(\rho)\big]R(\rho)\frac{d\rho}{\rho}
-2a\int_{-\alpha}^\alpha g(\varphi)\sin\left(\frac{\pi\varphi}{2\alpha}\right) d\varphi.
\label{defStilde}
\end{equation}
Indeed, each of the brackets of the RHS of \eqref{sol_Lap} can be written in the form
\begin{equation}
 F(R(r),T)=\ln\left[ \frac{R^2(r)}4-TR(r)+O(1)\right], \qquad r\to\infty,
\nonumber
\end{equation}
where $T$ is given, respectively, for each bracket by
\begin{equation}
 \frac12\sin\left(\frac{\pi\theta}{2\alpha}\right)R(\rho),\quad -\frac12\sin\left(\frac{\pi\theta}{2\alpha}\right)R(\rho),
\quad \sin\left(\frac{\pi\theta}{2\alpha}\right)\sin\left(\frac{\pi\varphi}{2\alpha}\right).
\nonumber
\end{equation}
As $r\to\infty$, we find
\begin{align}
 F(R(r),T)&=\ln\left\{ \frac{R^2(r)}4\left[1-4T\left(\frac{r}{a}\right)^{-\frac{\pi}{2\alpha}}
+O\left(\left(\frac{r}{a}\right)^{-\frac{\pi}{\alpha}}\right)\right]\right\}
\nonumber\\
&=\ln\left(\frac14\right)+{2}\ln\left[ R(r)\right]-4T\left(\frac{r}{a}\right)^{-\frac{\pi}{2\alpha}}
+O\left(\left(\frac{r}{a}\right)^{-\frac{\pi}{\alpha}}\right)
\nonumber\\
&=\ln\left(\frac14\right)+\frac{\pi}{\alpha}\ln\left(\frac{r}{a}\right)-4T\left(\frac{r}{a}\right)^{-\frac{\pi}{2\alpha}}
+O\left(\left(\frac{r}{a}\right)^{-\frac{\pi}{\alpha}}\right),\quad r\to\infty.
\label{FRasympt}
\end{align}

Using the estimate \eqref{FRasympt}, equation \eqref{sol_Lap} implies \eqref{qasympt}.

\end{rem}

\begin{rem}
 (A particular example) Suppose that we choose the functions $\{g_+(\rho),g_-(\rho),g(\varphi)\}$
as follows:
\begin{equation}
 g_{\pm}(\rho)=\big(\rho e^{\pm i\alpha})^s u\big(\rho e^{\pm i\alpha}\big),
\quad g(\varphi)=-\frac{i}{a}\big(a e^{i\varphi}\big)^s u\big(a e^{i\varphi}\big), \quad s\in\Complex.
\label{defgpm}
\end{equation}
Using these functions in the definition \eqref{defS} of $S$, we find
\begin{equation}
 S(s)=-\int_{H} z^{s-1}u(z) dz, \qquad z\in\Complex.
\label{Sint}
\end{equation}
Indeed, recall that $H$ involves three integrals; making in the first, second and third integrals of the RHS of \eqref{Sint} the substitutions, respectively,
\begin{equation}
 z=\rho e^{i\alpha}, \quad z=\rho e^{-i\alpha}, \quad z=a e^{i\varphi},
\nonumber
\end{equation}
we find that the RHS of \eqref{Sint} equals $-S$.

\end{rem}

\begin{rem}
 (The Riemann function) Suppose we chose the functions
$\{g_+(\rho), g_-(\rho), g(\varphi)\}$ by \eqref{defgpm},
where
\begin{equation}
 u(z)=\frac{1}{e^{-z}-1}.
\label{defu}
\end{equation}
Then the function $S(s)$ is proportional to the Riemann zeta function, namely
\begin{equation}
 S(s)=-\frac{2i\pi \zeta(s)}{\Gamma(1-s)}, \qquad z\in\Complex.
\label{defSa}
\end{equation}
Hence,the Riemann hypothesis is valid iff there does \emph{not} exist a solution of the Neumann boundary value problem defined in Theorem \ref{thm2.1} with the functions $\{g_+(\rho), g_-(\rho), g(\varphi)\}$ defined by equations \eqref{defgpm} and \eqref{defu}, which is bounded as $r\to\infty$.
\end{rem}

\medskip
\underline{\bf The Dirichlet Boundary Values}

\smallskip
By evaluating the RHS of equation \eqref{sol_Lap} at $r=a$ and at $\theta= \pm \alpha$ we find the following expressions for the Dirichlet boundary values:
\begin{equation}
\label{DBV1}
\begin{split}
q(a, \theta)=&-\frac{1}{\pi}\int_{a}^{\infty} \frac{g_+(\rho)}{\rho}\ln \left|\frac{1}{2}R(\rho)-\sin \left(\frac{\pi \theta}{2\alpha}\right)\right|d\rho+\frac{1}{\pi}\int_{a}^{\infty} \frac{g_-(\rho)}{\rho}\ln \left|\frac{1}{2}R(\rho)+\sin\left(\frac{\pi \theta}{2\alpha}\right)\right|d\rho, \\
&+\frac{a}{\pi}\int_{-\alpha}^{\alpha}g(\varphi)\ln\left|\sin\left(\frac{\pi \varphi}{2\alpha}\right)-\sin\left(\frac{\pi\theta}{2\alpha}\right)\right|d\varphi, \qquad -\alpha<\theta<\alpha,
\end{split}
\end{equation}
\begin{equation}
\label{DBV2}
\makeatletter
  \def\tagform@#1{\maketag@@@{$(#1)^\pm$\@@italiccorr}}
\makeatother
 \begin{split}
q(r,\pm \alpha) = &-\frac{1}{\pi} \int^\infty_a  g_+(\rho)\ln \left| \frac{R(r)}{2} \mp \frac{R(\rho)}{2}\right| \frac{d\rho}{\rho}+ \frac{1}{\pi} \int^\infty_a  g_-(\rho)\ln \left| \frac{R(r)}{2} \pm \frac{R(\rho)}{2}\right|\frac{d\rho}{\rho} \\
&+ \frac{a}{\pi} \int^\alpha_{-\alpha}  g(\varphi) \ln \left| \frac{R(r)}{2} \mp \sin \left(\frac{\pi\varphi}{2\alpha}\right) \right|d\varphi, \qquad  a<r<\infty.
 \end{split}
\end{equation}
Adding and substracing equations \eqref{DBV2}$^+$ and \eqref{DBV2}$^-$, we find the following equations:
\begin{align}
\label{add_DBV2}
q(r,-\alpha)+q(r,\alpha)=&
\frac{1}{\pi}\int_a^\infty\left[g_-(\rho)-g_+(\rho)\right]\ln \left|\frac{R^2(\rho)}{4}-\frac{R^2(r)}{4}\right|\frac{d\rho}{\rho}
\nonumber\\
&+\frac{a}{\pi}\int_{-\alpha}^{\alpha} g(\varphi)\ln \left|\frac{R^2(r)}{4}-\left(\sin \frac{\pi\varphi}{2\alpha}\right)^2\right|d\varphi
\end{align}
and
\begin{align}
\label{sub_DBV2}
q(r,-\alpha)-q(r,\alpha)=&
\frac{1}{\pi}\int_a^\infty\left[g_-(\rho)+g_+(\rho)\right]\ln \left|\frac{R(r)-R(\rho)}{R(r)+R(\rho)}\right|\frac{d\rho}{\rho}
\nonumber\\
&+\frac{a}{\pi}\int_{-\alpha}^{\alpha} g(\varphi)\ln \left|\frac{R(r)+2\sin\left( \frac{\pi\varphi}{2\alpha}\right)}{R(r)-2\sin\left(\frac{\pi\varphi}{2\alpha}\right)}\right| d\varphi.
\end{align}

In what follows we assume that the given functions $g_{\pm}(r)$ and $g(\theta)$ satisfy the constraint $S=0$.

Let $c$ be a constant, then
\begin{equation}
\ln \left |\frac{R(r)^2}{4}+c   \right| = G(r)+O(r^{-\pi/\alpha}), \quad r \rightarrow \infty,
\end{equation}
where $G(r)$ is defined by
\begin{equation}
 G(r)=\frac{\pi}{\alpha} \ln \left(\frac{r}{a} \right) -\ln 4, \quad r>0.
\end{equation}
Each of the two logarithmic terms in the RHS of equation \eqref{add_DBV2} grows logarithmically as $r \rightarrow \infty$, however the condition $S=0$, implies that these two terms cancel. Indeed, using the fact that $G(r)S=0$, equation \eqref{add_DBV2} can be rewritten in the form
\begin{align}
\label{add_DBV2_new}
q(r,-\alpha)+q(r,\alpha)=&
\frac{1}{\pi}\int_a^\infty\left[g_-(\rho)-g_+(\rho)\right]\left[ \ln \left|\frac{R^2(r)}{4}-\frac{R^2(\rho)}{4}\right|-G(r) \right]\frac{d\rho}{\rho}
\nonumber\\
&+\frac{a}{\pi}\int_{-\alpha}^{\alpha} g(\varphi)\left[ \ln \left|\frac{R^2(r)}{4}-\left(\sin \frac{\pi\varphi}{2\alpha}\right)^2\right|-G(r) \right] d\varphi.
\end{align}
Equations \eqref{sub_DBV2} and \eqref{add_DBV2_new} are the basic equations needed for the derivation
of certain integral identities.

\section{The Global Relation and Certain Integral Identities}
\label{s:3}

The global relations for the Laplace equation in the domain $D$ are the following two equations:
\begin{equation}
\label{GRs_Lap}
\makeatletter
  \def\tagform@#1{\maketag@@@{$(#1)^\pm$\@@italiccorr}}
\makeatother
\int_{\partial D} e^{\pm ik\theta} r^k \left[ (-q_\theta \pm ikq) \frac{dr}{r} + (rq_r-kq)d\theta\right] =0, \quad \Re k < \frac{\pi}{2\alpha}.
\end{equation}
The above restriction in $k$ is the consequence of the large $r$ behaviour of $q$, namely of the estimate
\begin{equation}
q = O\left(r^{-\frac{\pi}{2\alpha}} \right), \qquad r \rightarrow \infty.
\nonumber
\end{equation}
If the domain $D$ is defined by \eqref{D}, then
equations \eqref{GRs_Lap}$^\pm$ become
\begin{equation}
\label{GRs_Lap2}
\makeatletter
  \def\tagform@#1{\maketag@@@{$(#1)^\pm$\@@italiccorr}}
\makeatother
 \begin{split}
k \left\{ a^k \int^\alpha_{-\alpha} e^{\pm ik\theta} q(a,\theta)d\theta \pm i \int^\infty_a r^k \left[ e^{\mp i\alpha k} q(r,-\alpha) - e^{\pm i\alpha k} q(r,\alpha)\right]  \frac{dr}{r} \right\} \\
=a^{k+1} \int^\alpha_{-\alpha} e^{\pm ik\theta} g(\theta)d\theta + \int^\infty_a r^k \left[ e^{\mp i\alpha k} g_-(r) - e^{\pm i\alpha k} g_+(r)\right]  \frac{dr}{r}.
 \end{split}
\end{equation}
Adding and subtracting \eqref{GRs_Lap2}$^\pm$ we obtain the following basic equations:
\begin{equation}
\label{add_GRs}
 \begin{split}
k \left\{ a^k\int^\alpha_{-\alpha} \cos(k\theta) q(a,\theta)d\theta + \sin(\alpha k) \int_a^\infty r^k \left[ q(r,-\alpha) + q(r,\alpha)\right] \frac{dr}{r}\right\}  \\
=a^{k+1} \int^\alpha_{-\alpha} \cos (k\theta) g(\theta)d\theta + \cos (\alpha k) \int^\infty_a r^k [g_-(r) - g_+(r)]\frac{dr}{r}
 \end{split}
\end{equation}
and
\begin{equation}
\label{sub_GRs}
 \begin{split}
k \left\{ a^k \int^\alpha_{-\alpha} \sin (k \theta) q(a,\theta) d\theta + \cos (\alpha k) \int^\infty_a r^k  \left[ q(r,-\alpha) - q(r,\alpha)\right] \frac{dr}{r} \right\}  \\
=a^{k+1} \int^\alpha_{-\alpha} \sin (k\theta) g(\theta)d\theta - \sin (\alpha k) \int^\infty_a r^k \left[ g_-(r) + g_+(r)\right] \frac{dr}{r}.
 \end{split}
\end{equation}
Replacing in the RHS of the relations \eqref{add_GRs} and \eqref{sub_GRs} $q(r,-\alpha)\pm q(r,\alpha)$ by the RHS of \eqref{sub_DBV2} and \eqref{add_DBV2_new}, we find the following equations, which are valid for $\Re k<\frac{\alpha}{2\pi}$:
\begin{equation}
\label{F12}
a\int_{-\alpha}^{\alpha} g(\varphi)F_1(\varphi,k)d\varphi+\int_a^\infty\left[g_-(\rho)-g_+(\rho)\right]F_2(\rho,k)\frac{d\rho}{\rho}=0
\end{equation}
and
\begin{equation}
\label{F34}
a\int_{-\alpha}^{\alpha} g(\varphi)F_3(\varphi,k)d\varphi+\int_a^\infty\left[g_-(\rho)+g_+(\rho)\right]F_4(\rho,k)\frac{d\rho}{\rho}=0,
\end{equation}
where the functions $\{F_j\}_1^4$ are defined as follows:
\begin{align}
F_1(\varphi,k)=& k\int_{-\alpha}^{\alpha}\cos (k\theta)\ln \left|\sin \left(\frac{\pi\varphi}{2\alpha}\right)+\sin\left(\frac{\pi\theta}{2\alpha}\right)\right|d\theta
\nonumber\\
&+k\sin(\alpha k) \int_a^\infty\left(\frac{r}{a}\right)^k\left[\ln \left|\frac{R^2(r)}{4}-\left(\sin\frac{\pi\varphi}{2\alpha}\right)^2\right|-G(r) \right]\frac{dr}{r}
\nonumber\\
&-\pi\cos (k\varphi), \kern8em -\alpha<\varphi<\alpha, \quad \Re k < \frac{\pi}{2\alpha},
\label{defF1}
\end{align}
\begin{align}
\label{defF2}
F_2(\rho,k)=&k\int_{-\alpha}^{\alpha}\cos (k\theta)\ln \left|\frac{R(\rho)}{2}
+\sin\left(\frac{\pi\theta}{2\alpha}\right)\right|d\theta
\nonumber\\
&+k\sin(\alpha k) \int_a^\infty\left(\frac{r}{a}\right)^k\left[ \ln\left|\frac{R^2(r)}{4}-\frac{R^2(\rho)}{4}\right|-G(r) \right] \frac{dr}{r}
\nonumber\\
& -\pi\cos (\alpha k)\left(\frac{\rho}{a}\right)^k, \kern8em a<\rho<\infty,
\quad \Re k < \frac{\pi}{2\alpha},
\end{align}
\begin{align}
\label{defF3}
F_3(\varphi,k)=&-k\int_{-\alpha}^{\alpha}\sin (k\theta)\ln \left|\sin\left( \frac{\pi\varphi}{2\alpha}\right)+\sin\left(\frac{\pi\theta}{2\alpha}\right)\right|d\theta
\nonumber\\
&+k\cos (\alpha k) \int_a^\infty\left(\frac{r}{a}\right)^k\ln\left|\frac{R(r)+2\sin\left( \frac{\pi\varphi}{2\alpha}\right)}{R(r)-2\sin\left(\frac{\pi\varphi}{2\alpha}\right)}\right|\frac{dr}{r}
\nonumber\\
&-\pi\sin (k\varphi), \kern8em -\alpha<\varphi<\alpha, \quad \Re k < \frac{\pi}{2\alpha},
\end{align}
\begin{align}
\label{defF4}
F_4(\rho,k)=& k\int_{-\alpha}^{\alpha}\sin (k\theta)\ln \left|\frac{R(\rho)}{2}+\sin\left( \frac{\pi\theta}{2\alpha}\right)\right|d\theta+k\cos(\alpha k) \int_a^\infty\left(\frac{r}{a}\right)^k\ln\left|\frac{R(r)-R(\rho)}{R(r)+R(\rho)}\right|\frac{dr}{r}
\nonumber\\
&+\pi\sin (\alpha k)\left(\frac{\rho}{a}\right)^k, \kern8em a<\rho<\infty, \quad \Re k < \frac{\pi}{2\alpha},
\end{align}
and principal value integrals are assumed if needed.

The function $g_{-}(\rho)-g_{+}(\rho)$ is related to $g(\varphi)$ through the equation $S=0$, whereas the function $g_{-}(\rho)+g_{+}(\rho)$ is independent of $g(\varphi)$. Hence equations \eqref{F12} and \eqref{F34} imply  the following basic identities:
\begin{equation}
 F_1(\varphi, k)=f(k), \quad F_2(\rho,k)=f(k)
\label{F1F2eq}
\end{equation}
and
\begin{equation}
F_3(\varphi, k)=0, \quad F_4(\rho, k)=0,
\label{F3F4eq}
\end{equation}
where
\begin{equation*}
 a<\rho<\infty, \quad -\alpha<\varphi<\alpha, \quad \Re k<\frac{\pi}{2 \alpha}
\end{equation*}
and $f(k)$ is some function of $k$.

The above equations will be verified explicitly in the next section, where it will also be shown that the function $f(k)$ is given by
\begin{equation}
 f(k)=-\sin (\alpha k) \left[2 \ln2+\frac{\pi}{\alpha k} \right].
\end{equation}

\begin{rem}
($\Re k<0$) Using the estimates
\begin{align}
 &\ln\left| \frac{R^2(r)}{4}-\bigg(\sin\frac{\pi\varphi}{2\alpha}\bigg)^2\right|-G(r)=O(r^{-\pi/\alpha}),
\quad |\varphi|<\frac{\pi}{2},\quad r\to\infty
\nonumber\\
\noalign{\noindent and}
&\ln\left| \frac{R^2(r)}{4}-\frac{R^2(\rho)}{\alpha}\right|-G(r)=O(r^{-\pi/\alpha}),
\quad 0<\rho<\infty,\quad r\to\infty,
\nonumber
\end{align}
it follows that the first two integrals in the RHS of two equation in \eqref{defF1} involving the above terms are well defined for $\Re k<\pi/2\alpha$.
If $k$ satisfies the stronger restriction $\Re k<0$, then it is \emph{not} necessary to { subtract} the term involving $G(r)$; actually in this case we find
\begin{equation}
 k\sin(\alpha k)\int_a^\infty \bigg(\frac{r}{a}\bigg)^k G(r) \frac{dr}{r}=f(k),\quad \Re k<0.
\end{equation}

\end{rem}

\begin{rem}
\label{rem3.2}
 ($\Re k<0$) If $\Re k<0$, then
\begin{equation}
 \tilde F_1(\varphi,k)=\tilde F_2(\rho,k)=F_3(\varphi,k)=F_4(\rho,k)=0, \quad |\varphi|<\frac{\pi}2,
\quad 0<\rho<\infty, \quad \Re k<0,
\label{F1F4}
\end{equation}
 where $\tilde F_1$ and $\tilde F_2$ denote the expression obtained form $F_1$ and $F_2$ by neglecting
the terms involving $G(r)$. Equations \eqref{F1F4} can be obtained \emph{directly} as follows:
the function
\begin{equation}
 q(r,\theta)=\Re \left(z^k\right)=r^k\cos(k\theta), \qquad \Re k<0,
\label{solnq}
\end{equation}
is a solution of the Neumann boundary value problem with
\begin{equation}
 g_{\pm}(r)=\mp kr^k\sin(k\alpha), \qquad g(\theta)=k a^{k-1}\cos(k\theta).
\label{BCgpm}
\end{equation}
Substituting $q(r,\alpha)=r^k\cos(k\alpha)$ and the expressions \eqref{BCgpm} in equation \eqref{DBV2}$^+$, we find
\begin{align}
 r^k\cos(k\alpha)=&\frac{k\sin(k\alpha)}{\pi}\int_a^\infty \rho^k\ln\left|\frac{R^2(\rho)}{4}
-\frac{R^2(r)}{4}\right| \frac{d\rho}{\rho}
\nonumber\\
&+\frac{ka^k}{\pi}\int_{-\alpha}^{\alpha}\cos(k\varphi)
\ln\left| \frac{R(r)}2-\sin\left(\frac{\pi\varphi}{2\alpha}\right)\right| d\varphi, \quad \Re k<0.
\end{align}
Replacing $\varphi$ with $-\varphi$ and dividing by $a^k/\pi$, we find $\tilde F_2=0$.

Similarly, substituting $q(a,\theta)=a^k\cos(k\theta)$ and the expressions \eqref{BCgpm}
in equation \eqref{DBV1}, we find $\tilde F_1=0$.

The function
\begin{equation}
 q(r,\theta)=\Im \left(z^k\right)=r^k\sin(k\theta), \qquad \Re k<0,
\nonumber
\end{equation}
is also a solution of the Neumann boundary value problem with
\begin{equation}
 g_{\pm}(r)=\pm k r^k\sin(k\alpha), \qquad g(\theta)=k a^{k-1}\cos(k\theta).
\nonumber
\end{equation}
The above solution implies $F_3=F_4=0$.

\end{rem}

\begin{rem}
 
The first integral in the definition of $F_1$ involves a singularity at $\theta=-\varphi$. Making the change of variables $\theta=-\varphi+x$, the relevant singularity is mapped at $x=0$ and the associated integral is $\ln|x|$. This singularity can be handled by Cauchy principal value integrals. In this case the contribution of this singularity vanishes. Indeed, this contribution involves
\begin{equation}
 \lim_{\epsilon\to 0} \left( \int_{-\alpha}^{-\epsilon}+\int_{\epsilon}^{\alpha}\right)
\ln|x| dx=2\lim_{\epsilon\to 0}\int_{\epsilon}^{\alpha} \ln x dx
\nonumber
\end{equation}
and using integration by parts it follows that the contribution from $x=\epsilon$ vanishes.

If instead of the principal value integral we use the limit { from} above and below, the relevant contribution is 
\begin{equation}
 \lim_{\epsilon\to 0} i \int_0^\pi \ln(\epsilon) e^{i\theta} d\theta,
\nonumber
\end{equation}
which clearly does not exist.

\end{rem}

\section{Verification of the Four Identities}

In order to simplify the functions $\{F_j\}_1^4$ defined by \eqref{defF1}--\eqref{defF4}, we introduce the change of variables
\begin{align}
 &\frac{r}{a}=x^{-\frac{2\alpha}{\pi}}, \qquad \frac{\rho}{a}=y^{-\frac{2\alpha}{\pi}},
\nonumber\\
& k=-\frac{\pi}{2\alpha}\kappa,\qquad u=e^{\frac{i\pi\theta}{2\alpha}}, \qquad v=e^{\frac{i\pi\varphi}{2\alpha}}.
\label{changevar}
\end{align}
Then $\{F_j\}_1^4$ can be written as
\begin{align}
 F_1 &=-\kappa h_1(\kappa,v)+\kappa\sin\left(\frac{\pi\kappa}{2}\right)[h_3(\kappa,v)+h_3(\kappa,-v)]
-\frac{\pi}{2}(v^{\kappa}+v^{-\kappa}),\label{defF1h}\\
 F_2 &=-\kappa h_6(\kappa,y)+\kappa\sin\left(\frac{\pi\kappa}{2}\right)[h_4(\kappa,y)+h_5(\kappa,y)]
-\pi y^{\kappa}\cos\left(\frac{\pi\kappa}{2}\right),\label{defF2h}\\
 F_3 &=\kappa h_2(\kappa,v)+\kappa\cos\left(\frac{\pi\kappa}{2}\right)[h_3(\kappa,v)-h_3(\kappa,-v)]
+\frac{\pi}{2i}(v^{\kappa}-v^{-\kappa}),\label{defF3h}\\
F_4 &=-\kappa
h_7(\kappa,y)-\kappa\cos\left(\frac{\pi\kappa}{2}\right)[h_4(\kappa,y)-h_5(\kappa,y)]
-\pi y^{\kappa}\sin\left(\frac{\pi\kappa}{2}\right),\label{defF4h}
\end{align}
where the functions $\left\{h_j\right\}_1^7$ are defined by
\begin{align}
 h_1(\kappa,v)&=\int_{\gamma_1}\frac{u^{\kappa-1}+u^{-\kappa-1}}{2i}
\ln\left| \frac{v-v^{-1}}{2i}+\frac{u-u^{-1}}{2i}\right| du, \label{defh1}\\
h_2(\kappa,v) &=\int_{\gamma_1}\frac{u^{\kappa-1}-u^{-\kappa-1}}{2}
\ln\left| \frac{v-v^{-1}}{2i}+\frac{u-u^{-1}}{2i}\right| du, \label{defh2}\\
h_3(\kappa,v) &=\int_0^1 x^{\kappa-1}\left[ \ln\left|\frac{x+x^{-1}}{2}-\frac{v-v^{-1}}{2i}\right|
+\ln x+\ln2\right]dx, \label{defh3}\\
h_4(\kappa,y) &=\int_0^1 x^{\kappa-1}\left[ \ln\left|\frac{x+x^{-1}}{2}-\frac{y+y^{-1}}{2}\right|
+\ln x+\ln2\right]dx, \label{defh4}\\
h_5(\kappa,y) &=\int_0^1 x^{\kappa-1}\left[ \ln\left|\frac{x+x^{-1}}{2}+\frac{y+y^{-1}}{2}\right|
+\ln x+\ln2\right]dx, \label{defh5}\\
h_6(\kappa,y) &=\int_{\gamma_1}\frac{u^{\kappa-1}+u^{-\kappa-1}}{2i}
\ln\left| \frac{y+y^{-1}}{2}+\frac{u-u^{-1}}{2i}\right| du, \label{defh6}\\
h_7(\kappa,y) &=\int_{\gamma_1}\frac{u^{\kappa-1}-u^{-\kappa-1}}{2}
\ln\left| \frac{y+y^{-1}}{2}+\frac{u-u^{-1}}{2i}\right| du, \label{defh7}
\end{align}
and $\gamma_1$ is depicted in figure \ref{contour1}.

\medskip
In order to compute $\{h_j\}_1^7$, the following lemma will be
useful.

\begin{lem}
\label{lem4.1}
 Let $\gamma$ be a smooth curve from $z_1$ to $z_2$ in the $z-$complex plane which passes through $z_0$. Then
\begin{align}
I= PV\int_{\gamma}\frac{z^p}{z-z_0}dz=&
\frac{1}{(p+1)z_0}[z_1^{p+1}\;_2F_1(1,p+1;p+2;\frac{z_1}{z_0})
\nonumber\\
&-z_2^{p+1}\;_2F_1(1,p+1;p+2;\frac{z_2}{z_0})]-\pi i z_0^p, \qquad \Re p>-1.
\label{inthyper}
\end{align}
If the singularity lies off the path $\gamma$, then the last term, in  {(4.13)} is absent.
\end{lem}

\begin{proof}

Letting $x=z/z_0$, we find
\begin{equation}
I=z_0^p PV\int_{\gamma_0}\frac{x^p}{x-1}dx,
\nonumber
\end{equation}
where $\gamma_0$ runs from $z_1/z_0$ to $z_2/z_0$, avoids $x=0$ and passes through $x=1$.
We recall that
\begin{equation}
\frac{1}{p+1}\;_2F_1(1,p+1;p+2;x)=\int_0^1\frac{\rho^p}{1-\rho x}d\rho,\qquad
\Re p >-1.
\nonumber
\end{equation}
Letting $\rho x=v$, we find
\begin{equation}
\frac{x^{p+1}}{p+1}\;_2F_1(1,p+1;p+2;x)=\int_0^x\frac{v^p}{1-v}dv,
\nonumber
\end{equation}
so that
\begin{equation}
\frac{x^p}{x-1}=-\frac{1}{p+1}\frac{d}{dx}[x^{p+1}\;_2F_1(1,p+1;p+2;x)],
\end{equation}
which immediately gives \eqref{inthyper}, the principal part at $z_0$ being included by subtracting half the residue at this point.
\qedhere

\end{proof}

The hypergeometric functions in \eqref{inthyper} are related to the Lerch transcendent, a generalization of the Riemann Zeta function
\begin{equation}
\;_2F_1(1,p;p+1;z)=p\Phi(z,1,p);
\end{equation}
the Lerch transcendent is analytic in the $z$-plane for fixed $p\ne1$.

%
%%%%%%%%%%%%%%%%%%%%%%%%%%%%%%%%%%%%%%%%%%%%%%%%%%%%%%%%%%%%%%%%%%%%%%%%%%%%%%%%%%%%%
\begin{figure}[t!]
\kern\medskipamount
\centerline{
\includegraphics[width=0.315\textwidth]{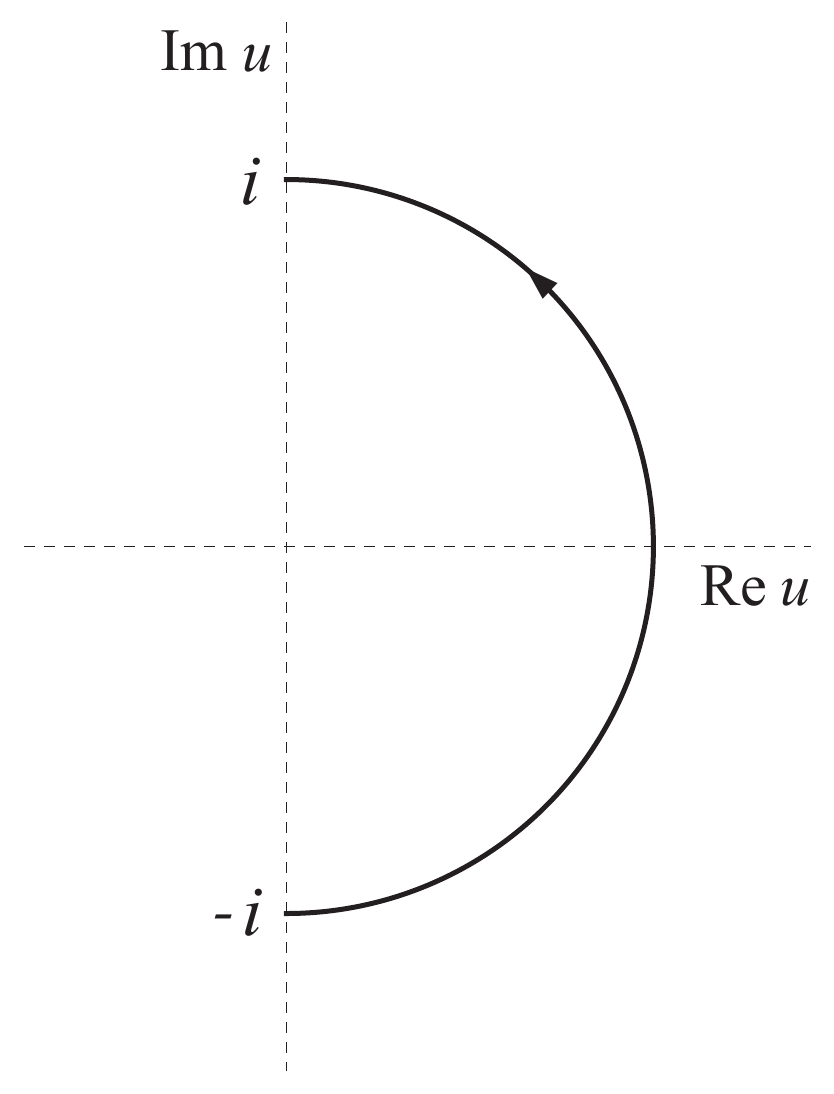}}
\caption{ The oriented contour $\gamma_1$. } \label{contour1}
\vskip3\medskipamount
\end{figure}
%%%%%%%%%%%%%%%%%%%%%%%%%%%%%%%%%%%%%%%%%%%%%%%%%%%%%%%%%%%%%%%%%%%%%%%%%%%%%%%%%%%%%
%

\begin{prop}
\label{prop4.1}
For $\Re\kappa >-1$,
\begin{equation}
 F_1(\varphi,k)=F_2(\rho,k)=-\sin(\alpha k)\left[2\ln2+\frac{\pi}{\alpha k}\right].
\label{F1F2}
\end{equation}
\end{prop}

\begin{proof}
If $\kappa=0$ (i.e. $k=0$), then $F_1=F_2=-\pi$, which {agrees} with \eqref{F1F2} in the limit $k\to 0$.
Thus, we consider the case of $\kappa\ne 0$.
Using integration by parts and a partial fraction decomposition, we find
\begin{align}
 h_1 (\kappa,v)=&\frac{1}{\kappa}\sin\left(\frac{\pi\kappa}{2}\right)\left[
\ln\left|\frac{v-v^{-1}}{2i}+1\right|+\ln\left|\frac{v-v^{-1}}{2i}-1\right|\right] 
\nonumber\\
&+\frac{1}{2i\kappa(v+v^{-1})}\,PV\int_{\gamma_1} \big(u^\kappa-u^{-\kappa}\big)\big(u+u^{-1}\big)
\left(\frac1{u+v}-\frac1{u-v^{-1}}\right) du,
\label{h1inta}
\end{align}
where the contour $\gamma_1$ is depicted in figure
\ref{contour1} and we have used
\begin{equation}
(i)^\kappa-(-i)^{\kappa}=2i\sin\frac{\pi\kappa}2.
\nonumber
\end{equation} Using the
change of variables $u\to 1/u$ for the term involving $u^{-\kappa}$,
the above principal value integral becomes
\begin{equation}
 \frac{1}{2i\kappa(v+v^{-1})} \,PV\int_{\gamma_1} u^\kappa\big(u+u^{-1}\big)
\left(\frac1{u+v}-\frac1{u-v^{-1}}-\frac1{u-v}+\frac1{u+v^{-1}}\right) du.
\label{h1intb}
\end{equation}
Each integral in \eqref{h1intb} can be computed by using lemma \ref{lem4.1}.
Noting that $-\alpha<\varphi<\alpha$, the definition of $v$ (see \eqref{changevar})
implies that $v$ and $v^{-1}$ lie on the contour $\gamma_1$.
The relevant residue contributions yield the term
\begin{equation}
 \frac{\pi}2 (v^{\kappa}+v^{-\kappa}).
\end{equation}
Hence, we find
\begin{align}
\label{h1intc}
\frac{1}{2i\kappa(v+v^{-1})}& \,PV\int_{\gamma_1} u^\kappa\big(u+u^{-1}\big)
\left(\frac1{u+v}-\frac1{u-v^{-1}}-\frac1{u-v}+\frac1{u+v^{-1}}\right) du
\nonumber\\
&= -\frac{1}{\kappa(v+v^{-1})}\sin\left(\frac{\pi\kappa}2\right)\left[ \frac1v \tilde F(\kappa+1;-i/v)
+ \frac1v \tilde F(\kappa+1;i/v)
+v\tilde F(\kappa+1;-iv) \right.
\nonumber\\
&\left. +v\tilde F(\kappa+1;iv) \right.
\left. -\frac1v \tilde F(\kappa-1;-i/v)- \frac1v \tilde F(\kappa-1;i/v)
-v\tilde F(\kappa-1;-iv)-v\tilde F(\kappa-1;iv)\right]
\nonumber\\
&-\frac{\pi}{2\kappa} (v^{\kappa}+v^{-\kappa}),
\end{align}
where
\begin{equation}
 \tilde F(\kappa;u_0)=\frac1{\kappa+1}\;_2F_1(1,\kappa+1;\kappa+2;u_0).
\end{equation}

Regarding $h_3(\kappa,v)$, using integration by parts and a partial fraction decomposition, we find
\begin{align}
 h_3(\kappa,v)=&\frac{1}{\kappa}\left\{
\ln\left|\frac{v-v^{-1}}{2i}-1\right|-\frac{1}{\kappa}+\ln 2\right.
\nonumber\\
&\left. -\frac{i}{v+v^{-1}}\int_0^1 x^\kappa(x-x^{-1})\left(
\frac1{x+iv}-\frac1{x-iv^{-1}}\right) dx\right\}.
\label{h3inta}
\end{align}
Each integral of the second line in \eqref{h3inta} can be computed using lemma \ref{lem4.1}.
%\begin{align}
% \int_0^1 \frac{x^{\kappa}}{x-ix_0} dx=\frac{i}{(\kappa+1)x_0}\;_2F_1(1,\kappa+1;\kappa+2;-i/x_0)
%,\qquad \Re\kappa >-1.
%\label{inthyper2}
%\end{align}
The definition of $v$ implies that the integrand in \eqref{h3inta} does not have any singularities on $[0,1]$. Hence, we find
\begin{align}
 h_3(\kappa,v)=&\frac{1}{\kappa}\left\{
\ln\left|\frac{v-v^{-1}}{2i}-1\right|-\frac{1}{\kappa}+\ln 2 -\frac{1}{v+v^{-1}}\bigg[ \frac1v \tilde F(\kappa+1;i/v)-\frac1v\tilde F(\kappa-1;i/v)\right.
\nonumber\\
&\left.
+v\tilde F(\kappa+1;-iv)-v\tilde F(\kappa-1;-iv)\bigg] \right\}.
\label{h3int}
\end{align}

Substituting \eqref{h1intc} into \eqref{h1inta} and using
\eqref{h3int}, we find that the terms involving the
hypergeometric functions and the logarithmic functions cancel.
Hence, recalling $\kappa=-\frac{2\alpha}{\pi}k$, we obtain
\eqref{F1F2} for $F_1$.

%
%%%%%%%%%%%%%%%%%%%%%%%%%%%%%%%%%%%%%%%%%%%%%%%%%%%%%%%%%%%%%%%%%%%%%%%%%%%%%%%%%%%%%
\begin{figure}[t!]
\kern\medskipamount \centerline{
\includegraphics[width=0.255\textwidth]{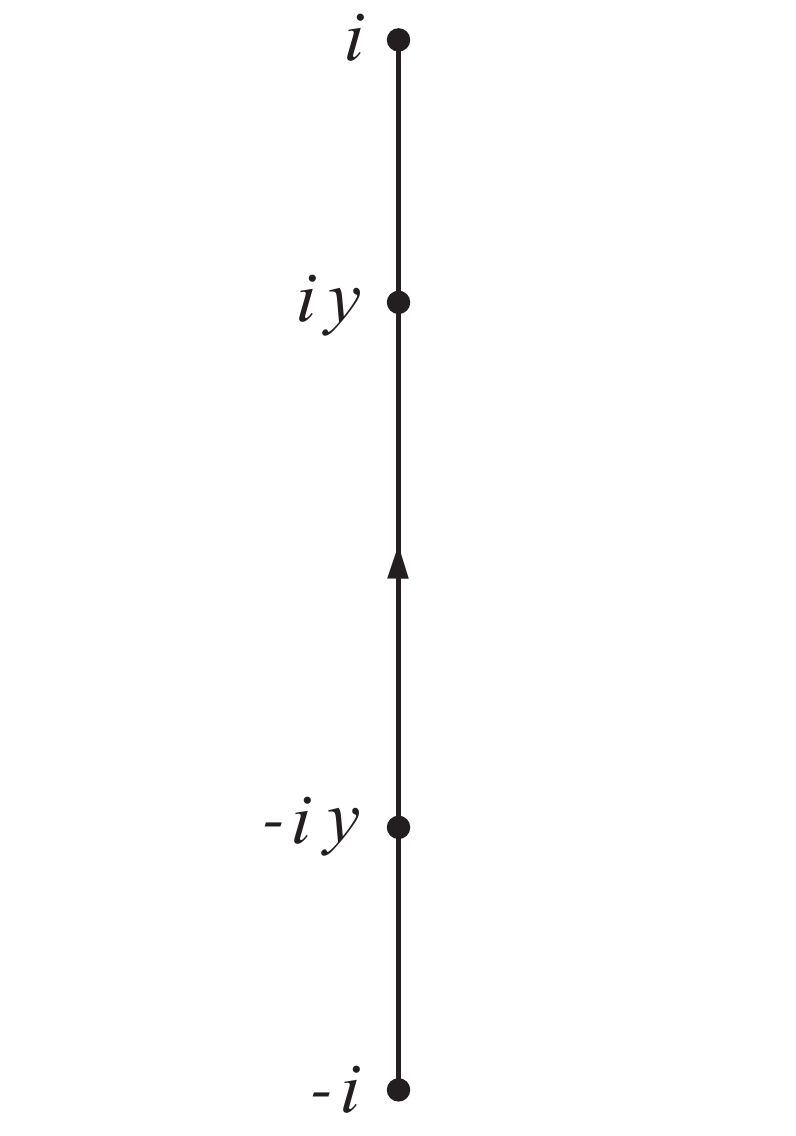}}
\caption{ The deformed contour 
$\gamma_2$. } \label{contour2} \vskip3\medskipamount
\end{figure}
%%%%%%%%%%%%%%%%%%%%%%%%%%%%%%%%%%%%%%%%%%%%%%%%%%%%%%%%%%%%%%%%%%%%%%%%%%%%%%%%%%%%%
%

Regarding $F_2$, we first deform the contour $\gamma_1$ to
$\gamma_2$, which is the segment $(-i,i)$ of the imaginary axis, see figure \ref{contour2}. Then,
using integration by parts and a partial fraction decomposition,
we find
\begin{align}
h_6(\kappa,y)=&\frac{1}{\kappa}\sin\left(\frac{\pi\kappa}{2}\right)\left[\ln\left|\frac12 (y+y^{-1})+1\right|
+\ln\left|\frac12 (y+y^{-1})-1\right| +i\pi y^{\kappa}\right]
\nonumber\\
&-\frac{1}{2\kappa(y-y^{-1})}\,PV\int_{\gamma_2}u^\kappa(u+u^{-1})
\left(\frac{1}{u+iy}-\frac{1}{u+iy^{-1}}-\frac{1}{u-iy}+\frac{1}{u-iy^{-1}}\right) du.
\label{h6inta}
\end{align}
Using the fact that the integrand of the above principal value integral has poles at $x=\pm iy$
and employing lemma \ref{lem4.1}, we find that the principal value integral of the RHS of \eqref{h6inta}
is given by
\begin{align}
 &\frac{1}{\kappa(y-y^{-1})}\sin\left(\frac{\pi\kappa}{2}\right)
\left[\frac1{y}\tilde F(\kappa+1,-1/y)-y\tilde F(\kappa+1,-y)+\frac1{y}\tilde F(\kappa+1,1/y)
-y\tilde F(\kappa+1,y)\right.
\nonumber\\
& \left.-\frac1{y}\tilde F(\kappa-1,-1/y)+y\tilde F(\kappa-1,-y)-\frac1{y}\tilde F(\kappa-1,1/y)
+y\tilde F(\kappa-1,y)\right]-\frac{\pi}{\kappa} y^{\kappa}\cos\left(\frac{\pi\kappa}{2}\right).
\label{h6intc}
\end{align}
For $h_4(\kappa,y)$ and $h_5(\kappa,y)$, using the fact that the integrands have a pole at $x=y$, we find
\begin{align}
 h_4(\kappa,v)=&\frac{1}{\kappa}\left\{
\ln\left|\frac{y+y^{-1}}{2}-1\right|-\frac{1}{\kappa}+\ln 2 +\frac{1}{y-y^{-1}}
\bigg[ \frac1y \tilde F(\kappa+1;1/y)-\frac1y\tilde F(\kappa-1;1/y)\right.
\nonumber\\
&\left.
-y\tilde F(\kappa+1;y)+y\tilde F(\kappa-1;y)\bigg] +i\pi y^{\kappa}\right\}
\label{h4int}
\end{align}
and
\begin{align}
 h_5(\kappa,v)=&\frac{1}{\kappa}\left\{
\ln\left|\frac{y+y^{-1}}{2}+1\right|-\frac{1}{\kappa}+\ln 2 +\frac{1}{y-y^{-1}}
\bigg[ \frac1y \tilde F(\kappa+1;-1/y)-\frac1y\tilde F(\kappa-1;-1/y)\right.
\nonumber\\
&\left.
-y\tilde F(\kappa+1;-y)+y\tilde F(\kappa-1;-y)\bigg] \right\}.
\label{h5int}
\end{align}

Combining \eqref{h6inta}, \eqref{h4int} and \eqref{h5int} with
\eqref{h6intc}, we find that the terms involving the
hypergeometric functions and the logarithmic functions cancel and
hence we find \eqref{F1F2} for $F_2$. \qedhere

\end{proof}

\begin{prop}
 \label{prop4.2}
For $\Re\kappa>-1$,
\begin{equation}
 F_3(\varphi,k)=F_4(\rho,k)=0.
\label{F2F3}
\end{equation}
\end{prop}

\begin{proof}
If $\kappa=0$, it is obvious that $F_3=F_4=0$. Thus, we consider
the case $\kappa\ne 0$. Using integration by parts and a partial
fraction decomposition, $h_2(\kappa,v)$ can be written as
\begin{align}
 h_2(\kappa,v)=&\frac{1}{\kappa}\cos\left(\frac{\pi\kappa}{2}\right)\left[
\ln\left|\frac{v-v^{-1}}{2i}+1\right|-\ln\left|\frac{v-v^{-1}}{2i}-1\right|\right]
\nonumber\\
& +\frac{1}{2\kappa(v+v^{-1})}\,PV\int_{\gamma_1} u^\kappa\big(u+u^{-1}\big)
\left(\frac1{u+v}-\frac1{u-v^{-1}}+\frac1{u-v}-\frac1{u+v^{-1}}\right) du ,
\label{h2int}
\end{align}
where we have used 
\begin{equation}
(i)^\kappa+(-i)^{\kappa}=2\cos\frac{\pi\kappa}2.
\nonumber
\end{equation}
The principal value integral in \eqref{h2int} can be evaluated by lemma \ref{lem4.1}:
\begin{align}
\label{h2inta}
 -\frac{1}{\kappa(v+v^{-1})}&\cos\left(\frac{\pi\kappa}2\right)\left[ \frac1v \tilde F(\kappa+1;-i/v)
- \frac1v \tilde F(\kappa+1;i/v)
-v\tilde F(\kappa+1;-iv)+v\tilde F(\kappa+1;iv) \right.
\nonumber\\
&\left. -\frac1v \tilde F(\kappa-1;-i/v)+\frac1v \tilde F(\kappa-1;i/v)
+v\tilde F(\kappa-1;-iv)-v\tilde F(\kappa-1;iv)\right]
\nonumber\\
&-\frac{\pi}{2i\kappa} (v^{\kappa}-v^{-\kappa}).
\end{align}
Substituting \eqref{h2inta} into \eqref{h2int} and combining the resulting expression
with \eqref{h3int}, we find $F_3=0$.

Regarding $F_4$, deforming the contour $\gamma_1$ to
$\gamma_2$ and then using integration by parts we find
\begin{align}
 h_7(\kappa,y)=&\frac{1}{\kappa}\cos\left(\frac{\pi\kappa}{2}\right)\left[
\ln\left|\frac12(y+y^{-1})+1\right|-\ln\left|\frac12(y+y^{-1})-1\right|-i\pi y^{\kappa}\right]
\nonumber\\
&+\frac{1}{2i\kappa(y-y^{-1})}\,PV\int_{\gamma_2} u^\kappa(u+u^{-1})
\left(\frac{1}{u+iy}-\frac{1}{u+iy^{-1}}+\frac{1}{u-iy}-\frac{1}{u-iy^{-1}}\right) du .
\label{h7int}
\end{align}
Using the fact that the integrand of the above principal value integral has poles at $x=\pm iy$ and
employing lemma \ref{lem4.1}, we find that the principal value integral of the RHS of \eqref{h7int} 
is given by
\begin{align}
 -&\frac{1}{\kappa(y-y^{-1})}\cos\left(\frac{\pi\kappa}{2}\right)
\left[\frac1{y}\tilde F(\kappa+1,-1/y)-y\tilde F(\kappa+1,-y)-\frac1{y}\tilde F(\kappa+1,1/y)
+y\tilde F(\kappa+1,y)\right.
\nonumber\\
& \left.-\frac1{y}\tilde F(\kappa-1,-1/y)+y\tilde F(\kappa-1,-y)+\frac1{y}\tilde F(\kappa-1,1/y)
-y\tilde F(\kappa-1,y)\right]-\frac{\pi}{\kappa} y^{\kappa}\sin\left(\frac{\pi\kappa}{2}\right).
\label{h7intc}
\end{align}
Using \eqref{h4int} and \eqref{h5int} together with \eqref{h7intc}, we find $F_4=0$.
\qedhere

\end{proof}

\section{Functional Identities}

The functions $h_3$, $h_5$, $h_6$ and $h_7$ do not involve principal value integrals, which always yields a term involving $i$, thus for $k$ real, it should be possible to express these functions in terms of real functions. We have succeeded in doing this for the last three functions but not for $h_3$. The relevant expressions are
\begin{align}
h_5(k,y)=&\frac{1}{k}\bigg\{\ln\bigg|\frac12(y+y^{-1})+1\bigg|-\frac1k+\ln2
\nonumber\\
&+\frac{1}{k(y^{-2}-1)}\big[y^{-2}\;_2F_1(1,k;k+1;-y^{-1})-\;_2F_1(1,k;k+1;-y)\big]
\nonumber\\
&-\frac{1}{(k+2)(y^{-2}-1)}\big[y^{-2}\;_2F_1(1,k+2;k+3;-y^{-1})-\;_2F_1(1,k+2;k+3;-y)\big]\bigg\},
\label{h5expr}\\
h_6(k,y)=&\frac{\sin\big(\frac{\pi k}{2}\big)}{k(y-y^{-1})^2}\bigg\{(y-y^{-1})^2
\ln\bigg|\frac12(y-y^{-1})\bigg|
-\frac{1}{k-2}\;_3F_2\bigg(\frac{1}{2},1,1;2-\frac{k}{2},\frac{k}{2};-\frac4{(y-y^{-1})^2}\bigg)
\nonumber\\
&+\frac{1}{k+2}\;_3F_2\bigg(\frac{1}{2},1,1;-\frac{k}{2},2+\frac{k}{2};-\frac4{(y-y^{-1})^2}\bigg)\bigg\},
\label{h6expr}
\\
h_7(k,y)=&-\frac{2}{k}\cos\bigg(\frac{\pi k}{2}\bigg)\bigg\{\frac12\bigg[\ln\bigg|\frac{y+y^{-1}}2-1\bigg|-\ln\bigg|\frac{y+y^{-1}}2+1\bigg|\bigg]
\nonumber\\
&+\frac{y+y^{-1}}{(y-y^{-1})^2}\bigg[\frac{1}{k+1} {\;_3F_2\bigg(\frac{1}{2},1, 1;\frac{1-k}{2},\frac{3+k}{2};
-\frac4{(y-y^{-1})^2}\bigg)}
\nonumber\\
&-\frac{1}{k-1}\;_3F_2\bigg(\frac{1}{2},1,1;\frac{3-k}{2},\frac{1+k}{2};-\frac4{(y-y^{-1})^2}
\bigg)\bigg]\bigg\}, \qquad \Re k >-1, \quad \ln y <0.
\label{h7expr}
\end{align}
In addition, for $k$ real, the functions $h_1$ and $h_2$ possess the following more complicated real forms:
\begin{align}
h_1(k,v)=&\frac{1}{k}\bigg\{2\sin\bigg(\frac{\pi k}{2}\bigg)
\ln\bigg|\frac{v+v^{-1}}2\bigg|+\frac{\pi^{3/2}}{(v+v^{-1})^2}\bigg[G_{33}^{21}
\left(\frac{(v+v^{-1})^2}4\bigg|\begin{array}{ccc}
1&-\frac{k}{2}&2+\frac{k}{2}\\
1&1&\frac{1}{2}\end{array}\right)
\nonumber\\
& -G_{33}^{21}\left(\frac{(v+v^{-1})^2}4\bigg|\begin{array}{ccc}
1&2-\frac{k}{2}&\frac{k}{2}\\
1&1&\frac{1}{2}\end{array}\right)\bigg]\bigg\}
\label{h1expr}\\
\noalign{\noindent and}
h_2(k,v)=&-\frac{1}{k}\bigg\{\cos\bigg(\frac{\pi k}{2}\bigg)\bigg[
\ln\left|\frac{v-v^{-1}}{2i}-1 \right|-\ln\left|\frac{v-v^{-1}}{2i}+1 \right| \bigg]
\nonumber\\
&-\frac{2\pi^{3/2}(v-v^{-1})}{i(v+v^{-1})} \bigg[G_{33}^{21}\left(\frac{(v+v^{-1})^2}4\bigg|
\begin{array}{ccc}
1&\frac{1-k}{2}&\frac{3+k}{2}\\
1&1&\frac{1}{2}\end{array}\right)
\nonumber\\
& +G_{33}^{21}\left(\frac{(v+v^{-1})^2}4\bigg|\begin{array}{ccc}
1&\frac{3-k}{2}&\frac{1+k}{2}\\
1&1&\frac{1}{2}\end{array}\right)\bigg]\bigg\},
\label{h2exprG}
\end{align}
where $G^{mn}_{pq}$ denotes the Meijer G-function.

The basic identities \eqref{F1F2eq} and \eqref{F3F4eq} yield a plethora of novel identities involving
the hypergeometric and related functions.

\medskip
\begin{exmp}
 
Replacing in the definition \eqref{defF4h} of $F_4$, $h_7$ by \eqref{h7expr}, $h_4$ by \eqref{h4int}
and $h_5$ by \eqref{h5expr}, the identity $F_4=0$ yields: 
\begin{align}
&{\frac{2}{(k+1)(y-y^{-1})^2} \;_3F_2\bigg(\frac{1}{2},1,1;\frac{1-k}{2},\frac{3+k}{2};-\frac{4}{(y-y^{-1})^2}\bigg)
+(k\rightarrow -k)}
\nonumber\\
&\kern1em={\frac{\pi y^{k}}{y+y^{-1}}\bigg(i+\tan\bigg(\frac{\pi k}{2}\bigg)\bigg)
+\frac{y}{y^2-y^{-2}}\bigg\{
\frac{1}{k}\bigg[\;_2F_1(1,k;1+k;y)-\;_2F_1(1,k;1+k;-y)\bigg]}
\nonumber\\
&\kern2em {-\frac{1}{k+2}\bigg[\;_2F_1(1,k+2;k+3;y)-\;_2F_1(1,k+2;k+3;-y)\bigg]\bigg\}}
\nonumber\\
&\kern2em +(y\rightarrow y^{-1}), \kern18em \Re k>-1,\qquad \ln y<0.
\label{ident1}
\end{align}
The $i$ appears above { as} a consequence of the fact that the hypergeometric function acquires an imaginary
part when continued to real arguments exceeding unity.

Letting $k=2$ with $z=2/(y-y^{-1})$
equation \eqref{ident1} yields the following novel identity:

\begin{equation}
{\;_3F_2 \bigg(\frac{1}{2},1,1;-\frac{1}{2},\frac{5}{2};-z^2\bigg)
=\frac{3}{z^2}\bigg\{\frac{4+3z^2}{\sqrt{1+z^2}}\frac{
\sinh^{-1}z}{z} -4\bigg\}}
\end{equation}
which, by analytic continuation,  is valid for all  $z$.

For $k=1/2$, equation \eqref{ident1} yields the { curious} identity
\begin{align}
&{\bigg(\frac{4}{(y-y^{-1})^2}\bigg)\bigg[ \;_3F_2\bigg(\frac{1}{2},1,1;\frac{1}{4},\frac{7}{4};-\frac4{(y-y^{-1})^2}\bigg)
+3\;_3F_2\bigg(\frac{1}{2},1,1;\frac{3}{4},\frac{5}{4};-\frac{4}{(y-y^{-1})^2}\bigg)\bigg]}
\nonumber\\
&={ \frac{3(1+i)\pi y^{3/2}}{y^2+1}+6\frac{\sqrt{y}}{y^2+1}\bigg[y\bigg(\tanh^{-1}\frac{1}{\sqrt{y}}-\tan^{-1}\frac{1}{\sqrt{y}}\bigg)+\bigg(\tanh^{-1}\sqrt{y}-\tan^{-1}\sqrt{y}\bigg)\bigg]}
\nonumber\\
\end{align}
for $0<y<1$, { where since the left hand side depends only on $4/(y-y^{-1})^2$, so must the right hand side, which is not at all obvious.}

Similarly, for $k=4$, we find the novel identity
\begin{align}
\;_3F_2 &\bigg(\frac{1}{2},1,1;-\frac{3}{2},\frac{7}{2};-\frac4{(y-y^{-1})^2}\bigg)
\nonumber\\
&={\frac{5(y-y^{-1})^2}y\bigg\{ \frac{y^4+y^2+1}{2(y^2+1)}\sinh^{-1}\bigg(\frac{2y}{1-y^2}\bigg)+\frac{y^2+1}{2y^2(y^2+1)}\ln\bigg(\frac{1+y}{1-y}\bigg)-\frac{y^2+1}{y}-\frac{1}{3}y\bigg\}}
\nonumber\\
&\kern1em { y^2\le1.}
\end{align}

\end{exmp}

\medskip
\begin{exmp}

Replacing in the definition \eqref{defF2h} of $F_2$, $h_6$ by \eqref{h6expr},
$h_4$ by \eqref{h4int} and $h_5$ by \eqref{h5expr}, the identity $F_2=f(k)$ yields the following 
identity:
\begin{align}
&{\frac{4}{(k+2)(y-y^{-1})^2}}\;_3F_2\bigg(\frac{1}{2},1,1;-\frac{k}{2},2+\frac{1}{2}k;-\frac4{(y-y^{-1})^2} \bigg)+(k\rightarrow-k) 
\nonumber\\
& \kern1em={2\pi y^{k}}\bigg(i-\cot\bigg(\frac{\pi k}{2}\bigg)\bigg)
+\frac{2y^{-1}}{y^{-1}-y}\bigg\{\frac{1}{k}\Big[\;_2F_1(1,k;k+1;y^{-1})+\;_2F_1(1,k;k+1;-y^{-1})\Big]
\nonumber\\
&\kern2em -\frac{1}{k+2}\Big[\;_2F_1(1,k+2;k+3;y^{-1})+\;_2F_1(1,2+k;3+k;-y^{-1})\Big]\bigg\}
\nonumber\\
&\kern2em +(y\rightarrow y^{-1}), \kern18em \Re k>-1, \qquad{| \ln y|<\infty}.
\end{align}

\end{exmp}

\medskip

Replacing in the definition{ \eqref{defF3h} }of $F_4$, $h_2$ by \eqref{h2exprG},
we find the novel identity, { valid for $Re\,  k >-1$ and all real $v$,}
\begin{align}
G^{21}_{33} &\left(\frac{(v+v^{-1})^2}{4}\bigg|\begin{array}{ccc} 1&\frac{1+k}{2}&\frac{3-k}{2}\\ 1&1&\frac{1}{2}\end{array}\right)+(k\rightarrow-k)
\nonumber\\
&{=\frac{(v+v^{-1})^2}{i\pi^{1/2}(v-v^{-1})} \bigg\{\frac{2}{\pi}\cos\bigg(\frac{\pi k}{2}\bigg) \bigg[\frac{1}{k+2}\Re[\frac{2v}{v+v^{-1}}\;_2F_1(1,2+k;3+k;iv)-(v\rightarrow v^{-1}) \bigg]}
\nonumber\\
&{-\frac{1}{k+1}\Re\bigg[\frac{v^2-1}{i(v-v^{-1})}\; _2F_1(1,k+1;k+2;iv^{-1})-(v\rightarrow v^{-1}) \bigg]-\frac{i }{2}(v^k-v^{-k}) \bigg\}.}
\end{align}
In particular, for $k=0${ and all real $b$,}
\begin{equation}
 G^{21}_{33}\left(\cos^2b\bigg|\begin{array}{ccc} 1&\frac{1}{2}&\frac{3}{2}\\ 1&1&\frac{1}{2}\end{array}\right)={\frac{2}{\pi^{3/2}}\cos b \cot b \ln\left(\frac{\cos b}{1+\sin b}\right)}
\end{equation}

\medskip
The above identities appear new. Thus, it seems that by employing 
the global relation to the solution of certain boundary value problems, 
it is possible to construct new formulas in the area of special functions, 
although it is not clear how the form of these identities can be predicted in advance.

\section*{Acknowledgement}
ASF acknowledges partial support from the Guggenheim Memorial foundation, USA.
The authors are grateful to Eugene Shargorodsky for several careful suggestions and
in particular for his observation expressed in Remark \ref{rem3.2}.

%%%%%%%%%%%%%%%%%%%%%%%%%%%%%%%%%%%%%%%%%%%%%%%%%%%%%%%%%%%

\end{document}